\documentclass[11pt]{amsart}
\usepackage[letterpaper, margin=1in]{geometry}			
\usepackage{verbatim, upquote}
\usepackage[ruled,norelsize]{algorithm2e}
\usepackage{graphicx, color, amssymb, url}
\usepackage{hyperref} 

\newtheorem{theorem}{Theorem}[section]

\newtheorem{lemma*}{Lemma}

\def \R {\mathbb{R}}

\def \F {\mathbb{F}}

\def \U {\mathcal{U}}
\def \C {\mathcal{C}}
\def \D {\mathcal{D}}

\newcommand{\CF}{\mathrm{CF}}
\newcommand{\MP}{\mathrm{MP}}

\title{Neural Ideals in SageMath}
\author[E. Petersen]{Ethan Petersen}
\address[E. Petersen]{Department of Mathematics, Rose-Hulman Institute of Technology, Terre Haute, IN 47803}
\email{peterseo@rose-hulman.edu}
\author[N. Youngs]{Nora Youngs}
\address[N. Youngs]{Department of Mathematics and Statistics, Colby College, Waterville, ME 04901}
\email{nora.youngs@colby.edu}
\author[R. Kruse]{Ryan Kruse}
\address[R. Kruse]{Mathematics Department, Central College,  Pella, IA 50219}
\email{kruser1@central.edu}
\author[D. Miyata]{Dane Miyata}
\address[D. Miyata]{Mathematics Department,
Willamette University,
Salem, OR 97301}
\email{dmiyata@willamette.edu}
\author[R. Garcia]{Rebecca Garcia}
\author[L.~D. Garc\'{\i}a Puente]{Luis David Garc\'{\i}a Puente}
\address[R. Garcia and L.~D. Garc\'{\i}a Puente]{Department of Mathematics and Statistics,
Sam Houston State University,
Huntsville, TX 77341-2206}
\email{rgarcia@shsu.edu}
\email{lgarcia@shsu.edu}

\date{\today}

\begin{document}

\begin{abstract}
A major area in neuroscience research is the study of how the brain
processes spatial information. 
Neurons in the brain represent external stimuli via neural
codes. These codes often arise from stereotyped stimulus-response
maps, associating to each neuron a convex receptive field.  
An important problem consists in determining what stimulus space
features can be extracted directly from a neural code. 
The neural ideal is an algebraic object that  
encodes the full combinatorial data of a neural code. This ideal can
be expressed in a canonical form that directly translates to a minimal
description of the receptive field structure intrinsic to the code.    
In here, we describe a SageMath package that contains several
algorithms related to the canonical form of a neural ideal. 
\end{abstract}

\maketitle

\section{Introduction}\label{sec:intro} 

Due to many recent technological advances in neuroscience,
researchers' ability to collect neural data has increased
dramatically. With this comes a need for new methods to process and
understand this data. One major question faced by researchers is to
determine how the brain encodes spatial features of its environment
through patterns of neural activity, as with place cell codes
\cite{OD71}.  In the recent paper \cite{CIVY13},  
Curto et al. 
phrase this question as, ``What can be inferred about the underlying
stimulus space from neural activity alone?'' 
To answer this question, Curto et al.
introduced the \emph{neural ring} and a related \emph{neural ideal},
algebraic objects that encode the full combinatorial data of a neural code.
They further show that the neural ideal can be expressed in a \emph{canonical form}
 that directly translates to a minimal description of the receptive
 field structure intrinsic to the code. 

In this article we describe a SageMath \cite{sage} package that
implements several algorithms to compute the neural ideal and its
canonical form, featuring a new iterative algorithm which proves to be
much more efficient than the original canonical form algorithm
outlined by Curto et al. in \cite{CIVY13}. This package also provides
an algorithm  to compute the primary decomposition of a
pseudo-monomial ideal. Accompanying these functions are others that
calculate related objects, such as the Gr\"obner basis, Gr\"obner fan,
universal Gr\"obner basis, and the neural ideal itself.  
%
%
In Section \ref{background} we give a short introduction to the
algebraic geometry of neural codes. 
Section \ref{iterative} describes a new and improved algorithm to
compute the canonical form of a neural ideal from a neural
code. Finally, Section \ref{tutorial} provides a tour of the
functionality within our code, centered on the canonical form
algorithm.  
    
\section{Background}\label{background}
In this section we give a brief introduction to the neural ring and
the neural ideal of a neural code. A more thorough background,
including necessary theorems and proofs, can be found in
\cite{CIVY13}. 

\subsection{Neural Codes and the Neural Ideal}

  A {\it neural code} $\C \subset{\{0,1\}^n}$ is a set of binary
  strings that represent neural activity. A `$1$' represents a firing
  neuron, while a `$0$' represents an idle neuron. For example, the
  presence of the word $1011$ in a 4-neuron code would indicate an
  instance when neurons 1, 3, and 4 were firing but neuron 2 was
  not. Given a neural code $\C\subset\{0,1\}^n$, the corresponding
  {\it neural ideal} $J_\C$ in the ring $\F_2[x_1,\dots,x_n]$ is
  defined by the following set of generators. For any $v \in
  \{0,1\}^n$, consider the polynomial 

$$\rho_v = \prod_{i=1}^n (1-v_i-x_i) = \prod_{\{i|v_i=1\} } x_i
\prod_{\{j|v_j=0\} } (1-x_j). 
$$

Notice that $\rho_v(x)$ acts as a characteristic function for $v$,
since it satisfies $\rho_v(v) = 1$ and $\rho_v(x) = 0$ for any $x\neq
v \in \{0,1\}^n$. The \emph{neural ideal} $J_\C\subseteq
\F_2[x_1,\dots,x_n]$ associated to the neural code $\C$
is the ideal generated by the polynomials $\rho_v$ with $v \notin \C$, that is, 
\[J_\C=\langle \rho_v \mid v \notin \C\rangle. \]

\subsection{Realizations and the Canonical Form}

Many systems of neurons react to stimuli which have a natural
geographic association. One example is head direction cells in
rats, where each neuron responds to a preferred range of angles; these
stimuli come from the 1-dimensional set of possible angles for the
head. Another example is place cells (in rats), where
each neuron is associated to a place field or region of the rat's
2-dimensional environment.  In such a geographic setup, we would
assume that if two neurons are observed to fire together, then the
sets of stimuli for these neurons must overlap. The idea of a \emph{realization
  for a code} formalizes these notions.

Suppose $\mathcal{U} = \{U_1,\dots,U_n\}$ is a collection of open sets
with each $U_i \subset{X}$.  Here, $X\subset \R^n$ represents the
space of possible stimuli, and $U_i$ is the receptive field of the
$i^{th}$ neuron, the set of stimuli which will cause that neuron to
fire.  We say that $\U$ is a {\it realization} for a code $\C$, or
that $\C = \C(\U)$, if $$\C = \{v\in\{0,1\}^n \,\mid \,
(\bigcap_{v_i=1} U_i )\backslash \bigcup_{v_j=0} U_j\};$$ 

that is, $\C$ represents the set of regions defined by $\U$. It turns
out that by considering the ideal $J_\C$ for the code $\C$, we can
determine complete information about the interaction of the $U_i$ in
any realization $\U$ of $\C$. This is facilitated by the
\emph{canonical form} of the neural ideal $J_\C$. 

A polynomial $f \in \F_2[x_1,\dots,x_n]$ is a \textit{pseudo-monomial}
if $f$ has the form \[f=\prod_{i\in \sigma}x_i \prod_{j\in \tau}
(1-x_j),\] where $\sigma\cap \tau = \emptyset$.  
An ideal $J \subset \F_2[x_1,\dots,x_n]$ is a \textit{pseudo-monomial
  ideal} if $J$ can be generated by a set of finitely many
pseudo-monomials. 
Let $J \subset \F_2[x_1,\dots,x_n]$ be an ideal, and $f \in J$ a
pseudo-monomial. Then $f$ is a \textit{minimal pseudo-monomial} of $J$
if there is no other pseudo-monomial $g \in J$ with $\deg(g) <
\deg(f)$ such that $f=gh$ for some $h \in \F_2[x_1,\dots,x_n]$. 
The \emph{canonical form} of a  pseudo-monomial ideal $J$, denoted
$\CF(J)$, is the set of \textit{all} minimal pseudo-monomials of $J$.  

The set $\CF(J)$ is unique for any given pseudo-monomial ideal $J$ and
$J=\langle \CF(J)\rangle$. It is important to  note that even though
$\CF(J)$ is made up of minimal pseudo-monomials, it does not
necessarily imply that $\CF(J)$ is a minimal set of generators for
$J$. The neural ideal $J_\C$ for a neural code  $\C$ is a
pseudo-monomial ideal since $J_\C= \langle\rho_v \mid v \notin
\C\rangle$ and each of the $\rho_v$'s is a pseudo-monomial.  The
following theorem describes the set of relations on any realization
$\U$ of $\C$ which $\CF(J)$ provides. 

\begin{theorem}\label{theorem:CF}
Let $\C\subset \{0,1\}^n$ be a neural code, and let
$\U=\{U_1,\dots,U_n\}$ be any collection of open sets in a stimulus
space $X$ such that $\C=\C(\U)$.  
Given $\sigma \subseteq \{1,\dots, n\}$, let $U_{\sigma} = \bigcap_{i
  \in \sigma} U_{i}$. Then 
 the \emph{canonical form} of $J_\C$ is: 
\begin{gather*}
J_\C= \bigg\langle \left\{x_\sigma \:|\: \sigma \text{ is minimal w.r.t. } U_\sigma=\emptyset \right\},\\
\left\{x_\sigma\prod_{i \in \tau}(1-x_i)
  \:|\:\sigma,\tau\neq\emptyset,\: \sigma\cap\tau=\emptyset,\:
  U_\sigma\neq\emptyset,\: 
\bigcup_{i\in\tau}U_i\neq X, \text{ and $\sigma,\tau$ are minimal
  w.r.t. }U_{\sigma}\subseteq\bigcup_{i\in \tau} U_i\right\},\\ 
\left\{\prod_{i\in \tau}(1-x_i)\:|\: \text{ $\tau$ is minimal w.r.t. }
  X\subseteq\bigcup_{i\in \tau}U_i\right\} \bigg\rangle. 
\end{gather*}
\end{theorem} 

We call the above three (disjoint) sets of relations comprising
$\CF(J_\C)$ the minimal Type 1, Type 2 and Type 3 relations,
respectively. 
Since the canonical form is unique, by Theorem \ref{theorem:CF}, any
receptive field representation of the code $\C = \C(U)$ satisfies the
following relationships: 

\begin{itemize}
\item[Type 1:] $x_{\sigma}\in \CF(J_\C)$ implies $U_{\sigma} =
  \emptyset$, but all intersections $U_{\gamma}$ where $\gamma
  \subsetneq \sigma$ are non-empty.
\medskip
\item[Type 2:] $x_{\sigma}\prod_{i\in \tau} (1-x_i) \in \CF(J_\C)$
  implies $U_{\sigma} \subseteq \bigcup_{i\in \tau} U_i$, but no
  lower-order intersection is contained in $\bigcup_{i\in \tau}U_{i}$
  and all the $U_i$s are needed for $U_{\sigma} \subseteq
  \bigcup_{i\in \tau} U_i$.
\medskip
\item[Type 3:] $\prod_{i\in \tau} (1-x_i) \in \CF(J_\C)$ implies $X
  \subseteq \bigcup_{i\in \tau} U_i$, but $X \not\subseteq
  \bigcup_{i\in \gamma} U_i$ for any  $\gamma \subsetneq \tau$. 
\end{itemize}


%

\section{The Iterative Algorithm}\label{iterative}

In \cite{CIVY13}, the authors provided a first algorithm to obtain the
canonical form via the primary decomposition of the neural
ideal. Here we present an alternative algorithm.
Rather than using the primary decomposition, this algorithm begins
with the canonical form for a code consisting of a single
codeword, and iterates by adding the remaining codewords one by one
and adjusting the canonical form accordingly. We describe the process
for adding in a new codeword in Algorithm \ref{algo1}.  


\begin{algorithm}[!hbtp]
\DontPrintSemicolon
    \SetKwInOut{Input}{Input} 
    \SetKwInOut{Output}{Output}
    \Input{$\CF(J_\C)=\{f_1,\dots,f_k\}$, where $\C\subset\{0,1\}^n$ is
      a code on $n$ neurons, and a codeword $c \in \{0,1\}^n$} 
    \Output{$\CF(J_{\C\cup\{c\}})$}
    \Begin{
      $L \longleftarrow \{\}$,\
      $M \longleftarrow \{\}$,\
      $N \longleftarrow \{\}$\;
       \For{$x\longleftarrow 1$ \KwTo $k$}{
         \eIf{$f_i(c) = 0$}{$L \longleftarrow L \cup \{f_{i}\}$\;}
         {$M \longleftarrow M \cup \{f_{i}\}$\;}}
       \For{$f \in M$}{
         \For{$j \longleftarrow 1$ \KwTo $n$}{
           \If{$f(x_j-c_j)$ is not a multiple of an element of $L$
             \textbf{and} $(x_j-c_j-1) \nmid f$}{$N \longleftarrow N
             \cup \{f(x_j-c_j)\}$\;} }} 
      \Return{$L\cup N = \CF(J_{\C\cup\{c\}})$}}        
    \caption{Iterative step to update a given canonical form after adding a code word $c$}\label{algo1}
\end{algorithm}

In summary, each pseudo-monomial $f$ from $\CF(J_\C)$ for which
$f(c)=0$ is automatically in the new canonical form
$\CF(J_{\C\cup\{c\}})$.  For those pseudo-monomials $f\in \CF(J_\C)$
with $f(c)=1$, we consider the product of $f$ with all possible linear
terms $(x_j-c_j)$ (so this product will be $0$ when evaluated at $c$)
as a possible candidate for $\CF(J_\C)$; but we remove any such
products which are redundant. Here, a redundant pseudo-monomial is one
which is either a multiple of another already known to be in the
canonical form, or which is a multiple of a Boolean polynomial.  

Certainly, any polynomial $f$ output by this algorithm will have the
property that $f(v) = 0$ for all $v\in \C\cup\{c\}$.  A proof that
this algorithm outputs {\it exactly} $\CF(J_{\C\cup \{c\}})$ is found
in the Appendix. 

We have developed  SageMath code that computes the canonical form
using this iterative algorithm; see Section \ref{tutorial} for an
in-depth tutorial of our package.  We find that
this iterative algorithm performs substantially better than the
original algorithm from \cite{CIVY13}. 
We have also implemented our algorithm in Matlab \cite{githubcode}.  

Table \ref{algo_runtimes} displays some runtime statistics regarding
our iterative canonical form algorithm. These runtime statistics were
obtained by running our SageMath implementation on 100 randomly
generated sets of codewords for each dimension $n = 4,\dots,
10$. These computations were performed on SageMath 7.2 running on a
Macbook Pro with a 2.8 GHz Intel Core i7 processor and 16 GB of
memory. We performed a similar test for our implementation of the
original canonical form algorithm in \cite{CIVY13} and on the Matlab
implementation of our iterative method. However, even in dimension
$5$ the original algorithm performs poorly. In our tests, we found
several codes for which the original algorithm took hundreds or even
thousands of seconds to compute the canonical form. For example, the
iterative algorithm takes 0.01 seconds to compute the canonical form
of the code below, but the original method takes 1 hour and 8 minutes
to perform the same computation.   
\begin{align*}
&10000,\,
 10001,\,
 01011,\,
 01010,\,
 10010,\,
 01110,\,
 01101,\,
 01100,\,
 11111,\,\\
 &11010,\,
 11011,\,
 01000,\,
 01001,\,
 00111,\,
 00110,\,
 00001,\,
 00010,\,
 00011,\,
 00101.
\end{align*}


 We also found several codes on dimension $5$ for which the original
 algorithm halted due to lack of memory. One such code is  
\[01110,\,
 01111,\,
 11010,\,
 11100,\,
 11101,\,
 01011,\,
 01000,\,
 00110,\,
 00001,\,
 00011,\,
 10011,\,
 00100,\,
 00101,\,
 10111.\]
The iterative algorithm computes the canonical form of the previous
code in $0.0089$ seconds. In our Matlab implementation we also found
several examples in dimension $6$ for which the canonical form took
thousands of seconds to be computed.

\begin{table}[!hbtp]
\centering
\begin{tabular}{|l | l | l | l | l| l|}
\hline
Dimension & min & max & mean & median & std \\ \hline
4         & 0.000077    & 0.0034   & 0.0016 & 0.0018 & 0.00076 \\
5         & 0.000087    & 0.014  & 0.0076 & 0.0082 &  0.0034   \\
6         & 0.00012    & 0.108  & 0.049 & 0.051 & 0.024        \\
7         & 0.00012    & 0.621  & 0.298 & 0.323 & 0.135        \\
8         & 0.000097    & 4.011  & 1.964 & 2.276 & 1.036       \\
9         & 0.698    & 39.28  & 24.86 & 27.38 &   9.976      \\
10        & 0.229  & 350.5  & 237.45 & 271.3 & 87.1          \\ \hline
\end{tabular}
    \caption{Runtime statistics (in seconds) for the iterative
      CF algorithm in SageMath.}
\label{algo_runtimes}
\end{table}


Our code has many features beyond the computation of the canonical
form via the new iterative algorithm.  In the following tutorial, we
show how to use the code to compute any of the following: 

\begin{enumerate}
\item The neural ideal.
\item The canonical form for a neural ideal via the new iterative algorithm.
\item The canonical form for a neural ideal via the primary decomposition algorithm.
\item A tailored method to compute the primary decomposition of pseudo-monomial ideals.
\item Gr\"obner bases and the Gr\"obner fan of a neural ideal.
\item A method to test whether a code is a simplicial complex.
\end{enumerate}

\section{SageMath Tutorial}\label{tutorial}
We decided to develop our code on SageMath due to its open-source
nature, extensive functionality, and ease of use \cite{sage}. In this
tutorial, we'll begin with installing the \texttt{NeuralCodes}
package. Then, we'll walk through the most important functions of the
package.

\subsection{Installation}
We will assume that SageMath is properly installed on the system. In our
tutorial, the files, \texttt{iterative\textunderscore canonical.spyx},
\texttt{neuralcode.py} and \texttt{examples.py}, are downloaded in the
folder \texttt{NeuralIdeals}. Now, we can load the package by: 

\begin{verbatim}
sage: load("NeuralIdeals/iterative_canonical.spyx")
sage: load("NeuralIdeals/neuralcode.py")
sage: load("NeuralIdeals/examples.py")
\end{verbatim}

The first file contains the iterative algorithm in Cython, so loading
will compile the code. Note that they must be loaded in this order, as
the Cython file must be compiled for the tests in
\texttt{neuralcode.py} to run. The second file,
\texttt{neuralcode.py}, holds all of the code that we will be
demonstrating. The third, \texttt{examples.py} has some additional
examples that can be loaded with \texttt{sage: neuralcodes()}. We expect
to include our code in SageMath, but currently it can be easily obtained at
\url{https://github.com/e6-1/NeuralIdeals}. 
We also want to note that running the package in the free
SageMathCloud (\url{https://cloud.sagemath.com}) is even easier. One
has to create a new project, upload all three files above and
then simply run \verb|load('neuralcode.py')|.

\subsection{Examples}
First, we define a neural code: 

\begin{verbatim}
sage: neuralCode = NeuralCode(['001','010','110'])
\end{verbatim}

Now, we can perform a variety of useful operations. We can compute the neural ideal:

\begin{verbatim}
sage: neuralIdeal = neuralCode.neural_ideal() 
sage: neuralIdeal
Ideal (x0*x1*x2 + x0*x1 + x0*x2 + x0 + x1*x2 + x1 + x2 + 1, x0*x1*x2 + x1*x2, 
x0*x1*x2 + x0*x1 + x0*x2 + x0, x0*x1*x2 + x0*x2, x0*x1*x2) of Multivariate 
Polynomial Ring in x0, x1, x2 over Finite Field of size 2
\end{verbatim}

We can compute the primary decomposition using a custom algorithm:

\begin{verbatim}
sage: pm_primary_decomposition(neuralIdeal)
[Ideal (x2 + 1, x1, x0) of Multivariate Polynomial Ring in x0, x1, x2 
  over Finite Field of size 2,
 Ideal (x2, x1 + 1) of Multivariate Polynomial Ring in x0, x1, x2 
  over Finite Field of size 2]
\end{verbatim}

We can compute the canonical form of the neural ideal.  

\begin{verbatim}
sage: canonicalForm = neuralCode.canonical()
sage: canonicalForm
Ideal (x1*x2, x1*x2 + x1 + x2 + 1, x0*x1 + x0, x0*x2) of
 Multivariate Polynomial Ring in x0, x1, x2 over Finite Field of size 2
\end{verbatim}

The method \texttt{canonical()} will use
the iterative algorithm by default. If we want to use the 
procedure described in \cite{CIVY13}, one needs to specify it with
\verb|canonical(algorithm="original")|. This procedure uses by default the tailored
pseudo-monomial primary decomposition, but one can make this explicit
with 
\verb|canonical(algorithm="original", decomposition_algorithm="pm")|.

\begin{verbatim}
sage: neuralCode.canonical(algorithm="original", decomposition_algorithm="pm")
Ideal (x1*x2, x0*x1 + x0, x1*x2 + x1 + x2 + 1, x0*x2) of 
 Multivariate Polynomial Ring in x0, x1, x2 over Finite Field of size 2
\end{verbatim}

Besides the tailored pseudo-monomial primary decomposition method, we
can also use the standard primary decomposition methods implemented in
SageMath such as Shimoyama-Yokoyama and  Gianni-Trager-Zacharias    
with the flags \verb|"sy"| and \verb|"gtz"|, respectively.
Table \ref{runtimes} compares the runtimes (in seconds) for the
primary decomposition of neural ideals using these three methods.  
Each entry in this table is the mean value of the runtimes for 50 randomly generated codes.

\begin{table}[!htbp]
	\centering
    \begin{tabular}{| l | l | l | l |}
    \hline
    Dimension & Pseudo-Monomial PD & Shimoyama-Yokoyama   & Gianni-Trager-Zacharias\\ \hline
    4         	& 0.16        & 0.028	& 0.027  \\
    5         	& 0.85        & 0.14	& 0.07 \\
    6           & 6.43         & 2.98	    & 0.27 \\
    7		& 63.1      &  74.9	    & 1.9      \\
    8           & 578 & 3040 & 45  \\ \hline 
    \end{tabular}
    \caption{Comparisons of runtimes for primary decompositions of neural ideals.}
 \label{runtimes}
\end{table}

We found that, in general, the custom pseudo-monomial primary decomposition
algorithm does not outperform the  Gianni-Trager-Zacharias
algorithm. Nevertheless, this procedure is implemented to be used in characteristic 0.
It also works in fields of large positive characteristic. But 
in small characteristic this procedure may not terminate.

We also want to observe that the \verb|canonical()| method returns an
Ideal object whose generators are not factored and hence not easy to
interpret in our context. In order to obtain the generators in the
canonical form in factored form, we use \verb|factored_canonical()|: 

\begin{verbatim}
sage: neuralCode.factored_canonical()
[x2 * x1, (x1 + 1) * x0, (x2 + 1) * (x1 + 1), x2 * x0]
\end{verbatim}

From this output we can easily read off the RF structure of the neural
code. However, we also provide a different command that parses the
output to explicitly describe the RF structure.  

\begin{verbatim}
sage: neuralCode.canonical_RF_structure()
Intersection of U_['2', '1'] is empty
X = Union of U_['2', '1']
Intersection of U_['0'] is a subset of Union of U_['1']
Intersection of U_['2', '0'] is empty
\end{verbatim}

We can also compute the Gr\"obner basis and the Gr\"obner fan of the
neural ideal. Note that we could compute the neural ideal and use the
built in \verb|groebner_basis()| method, but that approach will not
impose the conditions of the Boolean ring (i.e. $x^{2} + x = 0$). Our
method uses the built-in \verb|groebner_basis()| command but also
reduces modulo the Boolean equations. 

\begin{verbatim}
sage: neuralCode.groebner_basis()
Ideal (x0*x2, x1 + x2 + 1) of Multivariate Polynomial Ring in x0, x1, x2 
 over Finite Field of size 2
sage: neuralIdeal.groebner_basis()
[x0*x2, x1 + x2 + 1, x2^2 + x2]
\end{verbatim}

\begin{verbatim}
sage: neuralCode.groebner_fan()
[Ideal (x1 + x2 + 1, x0*x2) of Multivariate Polynomial Ring in x0, x1, x2 
over Finite Field of size 2]
\end{verbatim}

A neural code is called \emph{convex} if its codewords correspond to regions defined
by an arrangement of convex open sets in Euclidean space. Convex codes
have been observed experimentally in many brain areas. Hence there
has been increased interest in understanding what makes a neural code
convex \cite{convex}. It has also been observed that if a code is a
simplicial complex then it is convex. One can check if a code is a
simplicial complex with the command  \verb|is_simplicial(Codes)|: 

\begin{verbatim}
sage: is_simplicial(['001','010','110'])
False
sage: is_simplicial(['000','001','010','100','110','011','101','111'])
True 
\end{verbatim}

 



\appendix
\section*{Appendix: Proof of Iterative Algorithm}

Here, we show that the process described in Algorithm \ref{algo1}
gives $\CF(J_{\C\cup\{c\}})$ from $\CF(J_\C)$ and $c$. 
Throughout, we use the following conventions and terminology: $\C$ and
$\D$ are neural codes on the same number of neurons; so, $\C,\D
\subseteq \{0,1\}^n$. A monomial $x^\alpha$ is {\it square-free} if
$\alpha_i \in \{1,0\}$ for all $i=1,\dots,n$. A polynomial is {\it
  square-free} if it can be written as the sum of square-free
monomials.  For example: $x_1x_2+x_4+x_1x_3x_2$ is square-free.  There
is a unique square-free representative of every equivalence class of
$\F_2[x_1,\dots,x_n]/\langle x_i(1-x_i)\rangle$.  
For $h\in \F_2[x_1,\dots,x_n]$, let $h_R$ denote this unique square-free
representative of the equivalence class of $h$ in
$\F_2[x_1,\dots,x_n]/\langle x_i(1-x_i) \rangle$.  

Then, for $\CF(J_\C) = \{f_1,\dots,f_r\}$ and $\CF(J_\D) =
\{g_1,\dots,g_s\}$, we  define the set of {\it reduced
  products} $$P(\C,\D) \stackrel{\text{def}}{=} \{(f_ig_j)_R\,|\,
i\in [r], j\in [s]\}.$$ Note that as pseudo-monomials are square-free,
for each pair $i,j$ we have either $(f_ig_j)_R=0$ or $(f_ig_j)_R$ is a
multiple of both $f_i$ and $g_j$.  We define the {\it minimal reduced
  products} as $$\MP(\C,\D) \stackrel{\text{def}}{=}   \{h\in
P(\C,\D)\,|\, h\neq 0\text{ and } h\neq fg\text{ for any } f\in
P(\C,\D), \deg g\geq1\}.$$  

\begin{lemma*}\label{lem:cfalgorithm} If $\C,\D\subset\{0,1\}^n$, then
  the canonical form of their union is given by the set of minimal
  reduced products from their canonical forms: $\CF(J_{\C\cup \D}) =
  MP(\C,\D)$. 
\end{lemma*}

\begin{proof}
First, we show $\MP(\C,\D)\subseteq J_{\C\cup\D}$.  For any $h\in
\MP(\C,\D)$, there is some $f_i\in \CF(J_\C)$ and $g_j\in \CF(J_\D)$ so
$h=(f_ig_j)_R$.  In particular, $h\in J_\C$ as $h$ is a multiple of
$f_i$, and $h\in J_\D$ as it is a multiple of $g_j$.  Thus $h(c) = 0$
for all $c\in \C\cup \D$, so $h\in J_{\C\cup \D}$. 

Suppose $h\in \CF(J_{\C\cup \D})$.  Then as $J_{\C\cup \D}\subset
J_{\C}$, there is some $f_i\in \CF(J_{\C})$ so that $h=h_1f_i$, and
likewise there is some $g_j\in \CF(J_{\D})$ so $h=h_2g_j$ where $h_1,
h_2$ are pseudo-monomials. Thus $h$ is a multiple of $(f_ig_j)_R$ and
hence is a multiple of some element of $\MP(\C,\D)$.  But as every
element of $\MP(\C,\D)$ is an element of $J_{\C\cup \D}$, and $h\in
\CF(J_{\C\cup \D})$, this means $h$ itself must actually be in
$\MP(\C,\D)$. Thus, $\CF(J_{\C\cup\D}) \subseteq \MP(\C,\D)$.  For the
reverse containment, suppose $h\in \MP(\C,\D)$; by the above, $h\in
J_{\C\cup\D}$ It is thus the multiple of some $f\in \CF(J_{\C\cup\D})$.
But we have shown that $f\in \MP(\C,\D)$, which contains no multiples.
So $h=f$ is in $\CF(J_{\C\cup\D})$. 
\end{proof}

\begin{proof}[\bf{Proof of Algorithm}]   Note that if $c\in \C$, then
  $L = \CF(J_\C)$, so the algorithm ends immediately and outputs
  $\CF(J_\C)$; we will generally assume $c\notin \C$.  

 To show that the algorithm produces the correct canonical form, we
 apply Lemma \ref{lem:cfalgorithm}; it suffices to show that the set
 $L\cup N$ is exactly $\MP(\C, \{c\})$. This requires that all products
 are considered, and that we remove exactly those which are multiples
 or other elements, or zeros.  Note that $\CF(J_{\{c\}}) = \{ x_i - c_i
 \, | \, i \in [n]\}$. 

To see that all products are considered we will look at L and M
separately.  Let $g\in L$. Since $g(c) = 0$, we know
$(g\cdot(x_i-c_i))_R=g$ for at least one $i$.  So $g\in \MP(\C,\{c\})$.
Any other product $(g\cdot(x_j-c_j))_R$ will either be $0$, $g$, or a
multiple of $g$, and hence will not appear in $\MP(\C,\{c\})$. Thus,
all products of linear terms with elements of $L$ are considered, and
all multiples or zeros are removed.   It is impossible for elements of
$L$ to be multiples of one another, as $L\subset \CF(\C)$. 

We also consider all products of elements of $M$ with the linear
elements of $\CF(J_{\{c\}})$. We discard them if their reduction would
be 0, or if they are a multiple of anything in $L$.  If neither holds,
we add them to $N$.  So it only remains to show that no element of $N$
can be a multiple of any other element in $N$, and no element of $N$
can be a multiple of anything in $L$, and thus that we have removed
all possible multiples.  First, no element of $N$ may be a multiple of
an element of $L$, since if $g\in L$, $f\cdot(x_i-c_i)\in N$, and
$f\cdot(x_i-c_i)\cdot p = g$ for some pseudo-monomial $p$, then
$f\big| g$. But this is impossible as $f,g$ are both in $\CF(J_\C)$.
Now, suppose $f\cdot(x_i-c_i)= h\cdot g\cdot(x_j-c_j)$ for $f,g\in
\CF(J_\C)$ and $f\cdot(x_i-c_i), g\cdot(x_j-c_j)\in N$, and $h$ a
pseudo-monomial.  Then as $f\not|g$ and $g\not | f$, we have $i\neq
j$, and so $(x_j-c_j)\big|f$. But this means $f\cdot(x_j-c_j)=f$ and
therefore $f\in L$, which is a contradiction.  So no elements of $N$
may be multiples of one another. 

\end{proof}

%

\end{document}